\documentclass{llncs}
\IfFileExists{pdfsync.sty}{\usepackage{pdfsync}}{}
\usepackage{amsmath}
\usepackage{amsfonts}

\newcommand\motnouv[1]{\emph{#1}}

\newcommand\finprovisoire{
\bibliographystyle{plain}
\bibliography{/Users/bournez/bibliographie/Bib-Files/bournez,/Users/bournez/bibliographie/Bib-Files/perso}
\end{document}
}

\begin{document}

\begin{frontmatter}
\title{
On the Convergence of Population Protocols When Population Goes to Infinity\thanks{Olivier Bournez and Johanne Cohen were partially supported  from ANR Project SOGEA. Olivier Bournez, Johanne Cohen and Xavier Koegler were partially supported  by ANR Project SHAMAN. Xavier Koegler was also partially supported  by ANR Project ALADDIN and by COST Action 295 DYNAMO. } }
\date{\today}



\institute{
{\sc LIX, Ecole Polytechnique}, 91128 Palaiseau Cedex, FRANCE \\
{Olivier.Bournez@lix.polytechnique.fr}
\and
{\sc IECN/UHP}, 
BP 239, 54506 Vandoeuvre-L\`es-Nancy Cedex, FRANCE \\
{\{chassain,Lucas.Gerin\}@iecn.u-nancy.fr}
\and
{\sc LORIA/INRIA-CNRS}, 615 Rue du Jardin Botanique,
54602 Villers-L\`es-Nancy, FRANCE \\
{Johanne.Cohen@loria.fr}
\and
{\sc LIAFA/Paris VII University}, Case 7014
75205 Paris Cedex 13, FRANCE \\
\and 
{\sc Ecole Normale Sup\'erieure},
45, rue d'Ulm, 
75230 Paris cedex 05, FRANCE \\
{koegler@clipper.ens.fr}
}
\author{Olivier Bournez\inst{1} and Philippe Chassaing\inst{2} and Johanne Cohen\inst{3} and Lucas Gerin\inst{2} and Xavier Koegler\inst{4,5}}

\maketitle

\begin{abstract}
  Population protocols have been introduced as a model of sensor
  networks consisting of very limited mobile agents with no control
  over their own movement. A population protocol corresponds to a
  collection of anonymous agents, modeled by finite automata, that
  interact with one another to carry out computations, by updating
  their states, using some rules.

  Their computational power has been investigated under several
  hypotheses but always when restricted to finite size populations. In
  particular, predicates stably computable in the original model have been
  characterized as those definable in Presburger arithmetic.

  We study mathematically the convergence of  population
  protocols when the size of the population goes to infinity. We do so
  by giving general results, that we illustrate through the example of
  a particular population protocol for which we even obtain an
  asymptotic development.

  This example shows in particular that these
  protocols seem to have a rather different computational power
  when a huge population hypothesis is considered.
\end{abstract}

\end{frontmatter}

\newtheorem{rremarque}{Remark}
\newtheorem{mremarque}{Metaremark}
\newtheorem{rtodo}{TO DO:}
\newcommand\N{\mathbb{N}}
\newcommand\R{\mathbb{R}}
\newcommand\mN{\mathcal{N}}

\section{Motivation}

The computational power of networks of finitely many anonymous
resource-limited mobile agents has been investigated in several recent
papers. In particular, the population protocol model, introduced in
\cite{AspnesADFP2004}, consists of a population of finite-state agents
that interact in pairs, where each interaction updates the state of
both participants according to a transition based on the previous
states of the participants. When all agents converge after some finite
time to a common value, this value represents the result of the
computation.

According to survey \cite{PopProtocolsEATCS}, the model corresponds to the following assumptions: anonymous
finite-state agents (the system consists of a large population of
indistinguishable finite-state agents), computations by direct
interactions (an interaction between two agents updates their states
according to a joint transition table), unpredictable interaction
patterns (the choice of interactions is made by an adversary, possibly
limited to pairing only agents in an interaction graph), distributed
input and outputs (the input to a population protocol is distributed
across the initial state of the entire population, similarly the
output is distributed to all agents), convergence rather than
termination (the agent's output are required to converge after some
 time to a common correct value).

 Typically, in the spirit of \cite{AspnesADFP2004} and following
 papers (see again \cite{PopProtocolsEATCS} for a survey), population
 protocols are assumed to (stably) compute predicates: a population
 protocol stably computes a predicate $\phi$, if for any possible
 input $x$ of $\phi$, whenever $\phi(x)$ is true all agents of the
 population eventually stabilize to a state corresponding to $1$, and
 whenever $\phi(x)$ is false, all agents of the population eventually
 stabilize to a state corresponding to $0$.

Predicates stably computable by population protocols in this sense have been
characterized as being precisely the semi-linear predicates, that is
to say those predicates on counts of input agents definable in
first-order Presburger arithmetic \cite{presburger:uvk}. Semilinearity
was shown to be sufficient in \cite{AspnesADFP2004} and necessary in
\cite{AngluinAE2006semilinear}.

Here, we  study a new variant: we assume the number of agents in the population to be huge, or 
close to infinity (we call this \textit{a huge population hypothesis}), and we don't
want to focus on protocols as predicate recognizers, but as computing
functions. We assume outputs to correspond to proportions, which are
clearly the analog of counts whenever the population is infinite or
close to infinity.

We do so by providing general results  that we illustrate by
considering a particular population protocol, that we prove to
converge to a fraction of $\frac{\sqrt{2}}{2}$ agents in a given
state whatever its initial state is. We hence show that some algebraic
irrational values can be computed in this sense. Then we show how the
reasoning behind the proof of convergence of this particular protocol
can be generalized to any such protocol to prove that its behaviour
can be analyzed through use of deterministic differential equations.
We also give an asymptotic development of the convergence in the case
of the protocol computing $\frac{\sqrt{2}}2$.

Our motivation is twofold. First, to prove formally that population
protocol with a huge population hypothesis can be mathematically studied
using population models and ordinary differential equations. Second,
to show that protocols considered with these two hypotheses (huge
population, computing functions and not only predicates), have a
rather different power.

We consider this work as a first step towards understanding which
numbers can be computed by such protocols. Whereas we prove that
$\frac{\sqrt{2}}{2}$ can be computed, and whereas this is easy to see
that computable numbers in this sense must be algebraic numbers of
$[0,1]$, we didn't succeed yet to characterize precisely computable
numbers.

In this more long term objective, the aim of this current work is
first to discuss in which sense one can say that these protocols
compute an irrational algebraic value such as $\frac{\sqrt{2}}{2}$,
and second to study mathematically formally the convergence.

Our discussion is organized as follows. In Section \ref{sec:related},
we discuss 
related
work. In Section \ref{sec:pp}, we recall the model of population
protocols. In Section \ref{sec:system}, we present a particular
population protocol and we explain in which sense we would like to say
that this protocol computes some irrational algebraic value, with a
huge population hypothesis. We do so first by some informal study,
that we justify mathematically in the following three sections. We 
use a general theorem about approximation of diffusions  presented in Section \ref{sec:thoe}. How this theorem
yields the proof of convergence is explained in Section \ref{sec:conv}. Then, in
Section \ref{sec:gene}, we show how  the method
used on this particular example can be extended to prove that the study of
such protocols reduces to the study of differential
systems. The following two sections deal with additional results that
can be given on restricted parts of the model. We prove in Section
\ref{sec:asympt} that this is even possible to use the same theorem to
go further and get an asymptotic development of the convergence on the
example of $\frac{\sqrt{2}}{2}$. Section \ref{sec: otheralg} gives,
without proof, some first results on the type of algebraic numbers
that can be computed in this model using only two possible states for
all agents. Finally, Section \ref{sec:concl} is devoted to a conclusion and a
discussion.

\section{Related Work}
\label{sec:related}

The population protocol model was inspired in part by the work by
Diamadi and Fischer on trust propagation in social networks
\cite{diamadi2001sgs}. The model   is introduced and 
motivated  in \cite{AspnesADFP2004}
by the study of sensor networks in which passive agents are
carried along by other entities. The canonical example given in this
latter paper is the example of sensors attached to a flock of birds.

So far, most of the works on population protocols have concentrated on
characterizing what predicates on the input configurations can be
stably computed in different variants of the models and under various
assumptions. In particular, in   \cite{AspnesADFP2004}, Angluin et al. proved that all semi-linear
predicates can be computed by the original model but left open the question of their exact
power. This was solved in \cite{AngluinAE2006semilinear}, where it has
been proved that no-more predicates can be computed.

Variants of the original model considered so far include restriction to one-way communications,
restriction to particular interaction graphs, random interactions,
self-stabilizing solutions, through population protocols, to classical
problems in distributed algorithmic, the analysis of
various kind of failures of agents an d their consequences, etc. See survey
\cite{PopProtocolsEATCS}.  

As far as we know, a huge population hypothesis in the sense of this
paper, has not been considered yet, except in very recent papers
\cite{ChatzigiannakisS08,ChatzigiannakisS08rr}. In these papers, the
authors also study the dynamics and stability of probabilistic
population protocols via an ordinary differential equation
approach. They introduce several continuous models inspired by population models
in mathematics and by population protocols. Some relations between
continuous models are proved. Results from theory of stability for
continuous systems are used to discuss the stability of protocols. The
models considered in \cite{ChatzigiannakisS08,ChatzigiannakisS08rr} are not provably connected, in these papers, to classical finite population protocols. We prove in this paper that
there is a strong relation between classical finite population
protocols and models given by ordinary differential equations, and
hence with the models of \cite{ChatzigiannakisS08,ChatzigiannakisS08rr}. We also prove that the macroscopic view given by ordinary differential equations and by population models
is provably the weak limit of classical
finite population protocols when the size of the population goes to
infinity. Compared with \cite{ChatzigiannakisS08,ChatzigiannakisS08rr}, we focus on computability by these models, whereas \cite{ChatzigiannakisS08,ChatzigiannakisS08rr} focus on stability questions, and do not discuss computability issues.


Notice that we assume that interactions happen in a probabilistic way,
according to some uniform law. In the original population protocol
model, only specific fairness hypotheses were assumed on possible
adversaries \cite{AspnesADFP2004}. When the size of the population
goes to infinity, uniform sampling of agents seems to us the most
natural way to extend the fairness hypothesis. This assumption is
consistent with the interpretation of agents as autonomous biological
entities moving at random. Notice that this notion of adversary has
already been investigated for finite state systems
\cite{PopProtocolsEATCS}.


The results proved in this paper can be considered as a macroscopic
abstraction of a system given by microscopic rules of evolutions. See
survey \cite{givon2004emd} for general discussions about extraction of
macroscopic dynamics.

Whereas the ordinary differential equation \eqref{odee1} can be
immediately abstracted in a physicist approach from the dynamic
\eqref{ode:dynamic}, the formal mathematical equivalence of the two
approaches is not so trivial, and is somehow a strong motivation of
this work.

Actually, these problems seem to arise in many macroscopic
justifications of models from their microscopic description in
experimental science. See for example the very instructive discussion
in \cite{Les06} about assumptions required for the justification of
the Lotka-Volterra (predator-prey) model of population dynamics. In
particular, observe that microscopic correlations must
be neglected (i.e.  $E[XY]=E[X] E[Y]$ is needed, where $E$ denotes the
expectation). With a rather similar hypothesis (here assuming
$E[p^2]=E[p]^2$), dynamic \eqref{odee1} is clear from rules
\eqref{ode:dynamic}. Somehow, we prove here that this last hypothesis is not
necessary for our system. 

The techniques used in this paper are based on weak convergence
techniques, introduced in \cite{stroock1979mdp}, relating a stochastic
differential equation (whose solutions are called diffusions) to  approximations
 by a family of Markov processes. 
More specifically, we use a weaker form
\cite{comets2006csm} of a theorem from \cite{stroock1979mdp}.
Refer also to \cite{durrett1996scp} for an introduction to these techniques.

\section{Population Protocols}
\label{sec:pp}

We now recall definitions from \cite{AspnesADFP2004}. A protocol is given by $(Q,\Sigma,\iota,\omega,\delta)$ with the
following components. $Q$ is a finite set of \motnouv{states}.
$\Sigma$ is a finite set of \motnouv{input symbols}.  $\iota: \Sigma
\to Q$ is the initial state mapping, and $\omega: Q \to \{0,1\}$ is
the individual output function. $\delta \subseteq Q^4$ is a joint
transition relation that describes how pairs of agents can
interact. Relation $\delta$ is sometimes described by listing all
possible interactions using the notation $(q_1,q_2) \to (q'_1,q'_2)$,
or even the notation $q_1q_2 \to q'_1 q'_2$, 
for $(q_1,q_2,q'_1,q'_2) \in \delta$ (with the convention that
$(q_1,q_2) \to (q_1,q_2)$ when no rule is specified with $(q_1,q_2)$
in the left hand side). 


Computations of a protocol proceed in the following way. The
computation takes place among $n$ \motnouv{agents}, where $n \ge 2$. A
\motnouv{configuration} of the system can be described by a vector of
all the agent's states. The state of each agent is an element of $Q$. Because agents
with the same states are indistinguishable, each configuration can be
summarized as an unordered multiset of states, and hence of elements
of $Q$.

Each agent is given initially some input value from $\Sigma$: Each agent's initial
state is determined by applying $\iota$ to its input value. This
determines the initial configuration of the population.

An execution of a protocol proceeds from the initial configuration by
interactions between pairs of agents. Suppose that two agents in state
$q_1$ and $q_2$ meet and have an interaction. They can change into
state $q'_1$ and $q'_2$ if $(q_1,q_2,q'_1,q'_2)$ is in the transition
relation $\delta$.  If $C$ and $C'$ are two configurations, we write
$C \to C'$ if $C'$ can be obtained from $C$ by a single interaction of
two agents: this means that $C$ contains two states $q_1$ and $q_2$ and $C'$ is
obtained by replacing $q_1$ and $q_2$ by $q'_1$ and $q'_2$ in $C$,
where $(q_1,q_2,q'_1,q'_2) \in \delta$. An \motnouv{execution} of the
protocol is an infinite sequence of configurations
$C_0,C_1,C_2,\cdots$, where $C_0$ is an initial configuration and $C_i
\to C_{i+1}$ for all $i\ge0$. An execution is \motnouv{fair} if for
all configurations $C$ that appears infinitely often in the execution,
if $C \to C'$ for some configuration $C'$, then $C'$ appears
infinitely often in the execution.

At any point during an execution, each agent's state determines its
output at that time. If the agent is in state $q$, its output value is
$\omega(q)$. The configuration output is $0$ (respectively $1$) if all
the individual outputs are $0$ (respectively $1$). If the individual
outputs are mixed $0$s and $1s$ then the output of the configuration
is undefined. 

Let $p$ be a predicate over multisets of elements of
$\Sigma$. Predicate $p$ can be considered as a function whose range is
$\{0,1\}$ and whose domain is the collection of these multisets. The predicate is said to be computed by the protocol if,  for every  multiset $I$, and
every fair execution that starts from the initial configuration
corresponding to $I$, the output value of every agent eventually
stabilizes to $p(I)$.

The following was proved in
\cite{AspnesADFP2004,AngluinAE2006semilinear}

\begin{theorem}[\cite{AspnesADFP2004,AngluinAE2006semilinear}] A
  predicate is computable in the population protocol model if and only
  if it is semilinear.
\end{theorem}

Recall that semilinear sets are known to correspond to predicates on
counts of input agents definable in first-order Presburger arithmetic
\cite{presburger:uvk}.

\section{A Simple Example}
\label{sec:system}

Consider the following population protocol, with $Q=\{+,-\}$, and the following 
joint transition relation.





\begin{equation} \label{ode:dynamic}
\left\{
\begin{array}{lll}
++ & \to & +- \\
+- & \to & ++\\
-+ & \to & ++ \\
-- & \to & +- \\
\end{array}
\right.
\end{equation}


Using previous (classical) definition, this protocol does not stably
compute anything. Indeed, if we put aside the special configuration where
all agents are in state $-$ which is immediately left in any next
round, any configuration is reachable from any configuration.

However, suppose that we want to discuss the limit of the proportion $p(k)$ of agents in 
state $+$ in the population at discrete time $k$. If $n_+(k)$ denotes
the number of agents in state $+$, and $n_-(k)=n-n_+(k)$ the number of agents in state
$-$, $$p(k)=\frac{n_+(k)}{n}.$$


From now on, we suppose that at each time step, two different agents are
sampled uniformly among the $n$ particles, independently from the past.
Since we are dealing with $n$ indistinguishable agents, the population protocol is completely
described by the number of agents in state $+$.
We are then reduced to determine the evolution of the Markov chain
$$
\left(p(k)\right)_{k\in \mathbb{N}}\in \left\{\frac{0}{n},\frac{1}{n},\dots,\frac{n}{n}\right\}.
$$
The above discussion ensures that $(p(k))$ is an irreducible Markov chain in
$\{\frac{1}{n},\dots,\frac{n}{n}\}$. 
Let us now compute the transition probabilities of this irreducible Markov chain.
We have
$$
p(k+1)-p(k)\in\{-1,1\}.
$$
Then, we have to determine for each $i=1,2,\dots,n$
\begin{align*}
\pi^{(n)}(\tfrac{i}{n}\to \tfrac{i-1}{n})&:=\mathbb{P}\left(p(k+1)=\frac{i-1}{n}\ |\ p(k)=\frac{i}{n} \right),\\
\pi^{(n)}(\tfrac{i}{n}\to\tfrac{i+1}{n})&:=\mathbb{P}\left(p(k+1)=\frac{i+1}{n}\ |\ p(k)=\frac{i}{n} \right).
\end{align*}
Assume that $p(k)=i/n$ : $p(k)$ decreases
only if the two agents  sampled are in state $+$. That is,
$$
\pi^{(n)}(\tfrac{i}{n}\to \tfrac{i-1}{n})=\frac{\binom{i}{2}}{\binom{n}{2}}=\frac{i(i-1)}{n(n-1)}.
$$
In any other case, $p(k)$ increases by one :
\begin{align*}
\pi^{(n)}(\tfrac{i}{n}\to \tfrac{i+1}{n})&=1-\pi^{(n)}(\tfrac{i}{n}\to \tfrac{i-1}{n})\\
&=1-\frac{i(i-1)}{n(n-1)}.
\end{align*}
A consequence of the ergodic theorem is that the chain $(p(k))$
admits a unique stationary distribution $\mu$. By definition, it is the only application
$$
\mu:\left\{\frac{1}{n},\dots,\frac{n}{n}\right\}\to [0,1]
$$
such that
\begin{enumerate}
\item $\sum_{i=1}^n  \mu(i/n)=1$.
\item $\mu$ satisfies the \emph{balance equation}, \emph{i.e.} for each $i$
$$
\mu(\frac{i}{n}) = 
\mu(\frac{i-1}{n}) \pi^{(n)}(\tfrac{i-1}{n}\to \tfrac{i}{n}) + \mu(\frac{i+1}{n}) \pi^{(n)}(\tfrac{i+1}{n}\to\tfrac{i}{n}).
$$

\end{enumerate}
We do not pay attention to the exact expression of $\mu$. We only notice that, as the unique solution
to a rational system, it is an element of $\mathbb{Q}^n$. Hence, its mean $\sum_i \mu(i/n)i/n$ is
a rational number, that we denote $p^{(n)}$.

A second consequence of the ergodic theorem is the following convergence :
$$
\frac{p(1) + p(2) + ... + p(k)}{k} \stackrel{k\to\infty}{\rightarrow} p^{(n)},\text{ almost surely.} 
$$
The purpose of the rest of the discussion is to show that however, when $n$ goes to infinity, the 
mean value of $p(k)$ converges to the irrational number $\sqrt{2}/2$.
To see why this can be expected,  
observe that $$\pi^{(n)}(\tfrac{i}{n}\to \tfrac{i-1}{n}) = p^2(k) \frac{n}{n-1} - p(k) \frac{1}{n-1},$$
and write

\begin{equation} \label{eq:unn}
  \begin{array}[h]{rcl}
\mathbb{E}[n_+(k+1)-n_+(k)\ |\ n_+(k)]
&= &  \pi^{(n)}(\tfrac{i}{n}\to \tfrac{i+1}{n}) -   \pi^{(n)}(\tfrac{i}{n}\to \tfrac{i-1}{n}) \\
& = &  1 - 2 \pi^{(n)}(\tfrac{i}{n}\to \tfrac{i-1}{n}) \\
& = &    1 - 2 p^2(k) \frac{n}{n-1} +  p(k) \frac{2}{n-1} \\
  \end{array}
\end{equation}


 From this, we can derive (not yet rigourously) the asymptotic behavior of $p(k)$. Take indeed 
$n$ large, so that the right-hand term is close to $1-2p(k)^2$. Now, when $k$ goes large, if there is some convergence of the mean proportion of $+$, the system must concentrate on 
configurations that does not create or destroy +, in mean. Thus, $1-2p(k)^2$ must vanish
and $p(k) \approx \sqrt{2}/2$.

Hence the remaining problem is to justify and discuss mathematically the convergence.




\section{A General Theorem about Approximation of Diffusions}
\label{sec:thoe}

We will use the following theorem from \cite{stroock1979mdp}. We use
here the formulation of it in \cite{comets2006csm} (Theorem 5.8 page
96).

Suppose that for all integers $n\ge 1$, we have an homogeneous Markov
chain $(Y_k^{(n)})$ in $\R^d$ with transition kernel $\pi^{(n)}(x,dy)$, meaning that the law of $Y_{k+1}^{(n)}$, conditioned on
$Y_0^{(n)},\cdots,Y_k^{(n)}$, depends only on $Y_{k}^{(n)}$ and is
given, for all Borelian $B$, by
$$P(Y_{k+1}^{(n)} \in B|Y_k^{(n)})=\pi^{(n)}(Y_k^{(n)},B),$$
almost surely.

Define for $x \in \R^d$, 
\begin{align*}
  b^{(n)}(x) & = n \int (y-x) \pi^{(n)}(x,dy), \\
a^{(n)}(x) &= n \int (y-x)  (y-x)^* \pi^{(n)}(x,dy), \\
K^{(n)}(x) &=  n \int (y-x)^3 \pi^{(n)}(x,dy), \\
\Delta_\epsilon^{(n)}(x) &= n \pi^{(n)}(x,B(x,\epsilon)^c),
\end{align*}
where $B(x,\epsilon)^c$ denotes the complement of the ball with radius  $\epsilon$, centered at $x$. In other words,
$$
b^{(n)}(x)=n\mathbb{E}_x [ (Y_1-x) ],
$$
and
$$
a^{(n)}(x) = n\mathbb{E}_x[(Y_1-x)(Y_1-x)^*]
$$
where $\mathbb{E}_x$ stands for ``expectation starting from $x$", that is,
$$
\mathbb{E}_x [ (Y_1-x) ]=\mathbb{E}[ (Y_1-x) | Y_0=x].
$$
The coefficients $b^{(n)}$ and $a^{(n)}$ can be interpreted as the instantaneous drift and the variance (or matrix of covariance) of $X^{(n)}$.

Define 
 $$X^{(n)}(t) = Y^{(n)}_{\lfloor nt \rfloor}  + (nt - \lfloor nt \rfloor) (Y^{(n)}_{\lfloor nt+1 \rfloor} - Y^{(n)}_{\lfloor nt \rfloor}).
  $$

\begin{theorem}[Theorem 5.8, page 96 of \cite{comets2006csm}] \label{th:theo1}
Suppose that there exist some continuous functions $a,b$, such that for all $R<+\infty$,
\begin{align*}
  \lim_ {n\to \infty} sup_{|x| \le R} |a^{(n)}(x)-a(x)| &= 0 \\
\lim_ {n\to \infty} sup_{|x| \le R} |b^{(n)}(x)-b(x)| &= 0\\ 
\lim_{n\to \infty} sup_{|x| \le R} \Delta_\epsilon^{(n)} &= 0, \forall \epsilon >0\\
\sup_{|x| \le R} K^{(n)}(x) &< \infty.
\end{align*}

With $\sigma$ a matrix such that $\sigma(x)\sigma^*(x)= a(x)$, $x \in \R^d$, we suppose that the stochastic differential equation 
\begin{equation}
\label{eq:diffusion}
dX(t) = b (X(t)) dt + \sigma(X(t)) dB(t), ~~~~ X(0)=x,
\end{equation}
has a unique weak solution for all $x$. This is in particular the case, if it admits a unique strong solution. 

Then for all sequences of initial conditions  $Y_0^{(n)} \to x$, the sequence of random processes $X^{(n)}$ \emph{converges in law} to the diffusion given by \eqref{eq:diffusion}. In other words, for all function $F: \mathcal{C}(\R^+,\R) \to \R$ bounded and continuous, one has $$\lim_{n \to \infty} E[F(X^{(n)})] = E[F(X)].$$
\end{theorem}


\section{Proving Convergence of Previous Simple Example}
\label{sec:conv}

Consider $Y_i^{(n)}$ as the homogeneous Markov chain corresponding to
$p(k)$, when $n$ is fixed. From previous discussions, $\pi^{(n)}(x,.)$
is a weighted sum of two Dirac that weight $x-\frac1n$ and
$x+\frac1n$, with respective probabilities $\pi_{-1}$ and $\pi_{+1}$, whenever $x$ is of type $\frac{i}{n}$ for some $i$.

From Equation \eqref{eq:unn} we have
\begin{equation} \label{eq:truc}
 E[p(k+1)-p(k) | p(k)] =  \frac1n(1 - 2 p(k)^2 \frac{n}{n-1}  + p(k) \frac2{n-1}), 
\end{equation}
which yields
$$b^{(n)}(x) = 1 - 2 p(k)^2 \frac{n}{n-1}  + p(k) \frac2{n-1},$$
when $x=i/n$. Now, clearly $(p(k+1)-p(k))^2= \frac{1}{n^2}$, and hence
$$a^{(n)}(x) = \frac1n,$$
when $x=i/n$. Taking $a(x)=0$ and $b(x)=1-2x^2$, we get
$$\lim_ {n\to \infty} sup_{|x| \le R} |a^{(n)}(x)-a(x)| = 0$$
$$\lim_ {n\to \infty} sup_{|x| \le R} |b^{(n)}(x)-b(x)| = 0$$
for all $R < +\infty$. Since the jumps of $Y^{(n)}$ are bounded in absolute value by $\frac1{ n}$, $\Delta_\epsilon^{(n)}$ is null, as soon as $\frac1{ n}$ is smaller than $\epsilon$, and so
$$\lim_{n\to \infty} sup_{|x| \le R} \Delta_\epsilon^{(n)} = 0, \forall \epsilon >0.$$
Finally,
$$\sup_{|x| \le R} K^{(n)}(x) < \infty$$ is easy to establish.

Now, (ordinary and deterministic) differential equation 
\begin{equation}
\label{odee1} 
dX(t) = (1-2X^2) dt
\end{equation} has a unique solution for any initial condition. 

It follows from above theorem that the sequence of random processes $X^{(n)}$ defined by
$$X^{(n)}(t) = Y^{(n)}_{\lfloor nt \rfloor}  + (nt - \lfloor nt \rfloor) (Y^{(n)}_{\lfloor nt+1 \rfloor} - Y^{(n)}_{\lfloor nt \rfloor})
  $$
  converges in law to the unique solution of differential equation
  \eqref{odee1}.

  Clearly, all solutions of ordinary differential equation
  $\eqref{odee1}$ converge to $\frac{\sqrt 2}{2}$.  Doing the change of variable $Z(t)=X(t)-\frac{\sqrt 2}{2}$, we get
\begin{equation}
\label{ode11}
dZ(t)= (-2 Z^2 + 2 \sqrt 2 Z) dt,
\end{equation}
that converges to $0$.

Coming back to $p(k)$ using definition of $X^{(n)}(t)$, we hence get
\begin{theorem}\label{Th : 3}
  We have for all $t$, $$p(\lfloor nt \rfloor) = \frac{\sqrt2}{2} +
  Z_n(t),$$ where $Z_n(t)$ converges in law when $n$ goes to infinity
  to the (deterministic) solution of ordinary differential
  \eqref{ode11}. Solutions of this ordinary differential equation go
  to $0$ at infinity.
\end{theorem}

This implies that $p(k)$ must converge to
$\frac{\sqrt2}{2}$ when $k$ and $n$ go to infinity.

\section{Generalization To General Population Protocols}\label{sec:gene}

We will now generalize the reasoning made on this particular example
in order to prove that the behaviour of any population protocol can be
approximated by a deterministic differential equation in a similar
way.

Transition rules of a population protocol are of the form:

$$ q~q' \rightarrow \delta_1(q,q')~\delta_2(q,q')$$ for all $(q_1, q_2)\in Q^2$.

As previously, we consider pairwise interactions between two agents
chosen randomly according to an uniform law in a population of size
$n$.

Let us define the Markov chain $Y_i^{(n)}$ corresponding to the vector
of $\mathbb{R}^{Q}$ whose components are the proportions of agents in
the different states and

 $$X^{(n)}(t) = Y^{(n)}_{\lfloor nt \rfloor}  + (nt - \lfloor nt \rfloor) (Y^{(n)}_{\lfloor nt+1 \rfloor} - Y^{(n)}_{\lfloor nt \rfloor}).
  $$

\begin{theorem}\label{th:generalisation}
  Le $b$ be the function defined by :

$$b(x) = \sum_{(q,q')\in Q²} x_q x_{q'} (-(e_q +e_q') +
e_{\delta_1(q,q')} + e_{\delta_2(q,q')}) $$ where $(e_q)_{q\in Q}$ is
the canonical base of $\mathbb{R}^Q$.

Then for all sequences of initial conditions  $Y_0^{(n)} \to x$, the
sequence of random processes $X^{(n)}$ converges in law to the
solution of the ordinary differential equation:

\begin{equation}
\label{eq:diffusion2}
dX(t) = b (X(t)) dt , ~~~~ X(0)=x,
\end{equation}
\end{theorem}

\begin{remark}
Ordinary differential equation \eqref{eq:diffusion2} corresponds to a degenerated stochastic differential equation. Being deterministic, we are sure
that it has a unique weak solution for all $x$.
\end{remark}

\begin{proof}
  $Y^{n}_i$ is of the form required by
  Theorem~\ref{th:theo1} with $\pi^{(n)}(x,.)$ being the sum of
  $5^{|Q|}$ Dirac : the variation of the proportion of agents in any
   given state belongs to $\{\frac{-2}{n}, \frac{-1}{n}, 0, \frac{1}{n},
  \frac{2}{n}, \}$ and the probabilities of any of these variations
  are clearly only dependant on the current state $x$.

Now let us define $a^{(n)}(x), b^{(n)}(x), K^{(n)}(x)$ and
$\Delta_{\epsilon}^{(n)}$ as in Theorem~\ref{th:theo1}. Let $R$ be
any finite non-negative real number.

As in the example above, since at any given time step at most two out
of $n$ agents change state, $\Delta^{n}_{\epsilon} =0$ if
$\epsilon>\frac{4}{n}$ and thus $$\lim_{n\to \infty} sup_{|x| \le R} \Delta_\epsilon^{(n)} = 0, \forall \epsilon >0.$$

$$\sup_{|x| \le R} K^{(n)}(x) < \infty$$ is also easy to establish.

Similarly $$ \forall x\in \mathbb{R}^{|Q|}, |x|\le R,  |a^{(n)}(x)|\leq \frac{4|Q|}{n}.$$ So if
we take $a(x) = 0$, we have $$ \lim_{n\rightarrow \infty} sup_{|x|\le
  R} |a^{(n)}(x) - a(x)| = 0.$$

If we write, for all $(q,q') \in Q^2, q \neq q'$, $$\Pi^{(n)}_{q,q'}(x) = x_q x_{q'}\frac{n}{n-1}$$

and 

$$\Pi^{(n)}_{q,q}(x) = x_q x_{q}\frac{n}{n-1} - \frac{x_q}{n-1}  .$$

Then $\Pi^{(n)}_{q,q'}(x)$ is exactly the probability of an encounter
between an agent in state $q$ and an agent in state $q'$ to happen
when the population is in configuration $x$.
We then have :

 $$b^{(n)}(x) = \sum_{(q,q')\in Q²} \Pi^{(n)}_{q,q'}(x) (-(e_q +e_q') +
e_{\delta_1(q,q')} + e_{\delta_2(q,q')}),$$

or

$$b^{(n)}(x) = \frac{n}{n-1} b(x) - \frac{1}{n-1}\sum_{q\in Q} x_q (-2e_q +
e_{\delta_1(q,q)} + e_{\delta_2(q,q)}).$$

Thus, finally, 

$$\lim_ {n\to \infty} sup_{|x| \le R} |b^{(n)}(x)-b(x)| = 0.$$

We can now conclude by Theorem~\ref{th:theo1}.
\end{proof}

This means that to understand the asymptotic behaviour of any such
protocol, we can study the associated differential equation. It is
also of interest to note that the function $b$ defined here is a
quadratic form over $\mathbb{R}^Q$.

\section{An Asymptotic Development of the Example Dynamic}
\label{sec:asympt}

It is actually possible to go further, at least in some cases like the
previous simple example, and prove the equivalent of a
central limit theorem, or if one prefers, to do an asymptotic
development of the convergence, in terms of stochastic processes. We
shall do so for the example dynamic used before using the same
notations as in previous sections. 

In our previous example, as $p(k)$ is expected to converge to $\frac{\sqrt2}{2}$, consider the
following change of variable:

$$Y^{(n)}(k) = \sqrt n(p(k) - \frac{\sqrt{2}}{2}).$$

The subtraction of $\frac{\sqrt2}{2}$ is here to get something
centered, and the $\sqrt n$ factor is here in analogy with classical
central limit theorem.

Clearly, $Y^{(n)}(.)$, that we will also note $Y(.)$ in what follows
when $n$ is fixed, is still an homogeneous Markov Chain.

We have
$$E[ Y(k+1)- Y(k) | Y(k)] = \sqrt n (E[ p(k+1)- p(k) | p(k)]),$$
hence, from \eqref{eq:truc},
$$E[ Y(k+1)- Y(k) | Y(k)] = \frac{1}{\sqrt n} ( 1 - 2 p(k)^2 \frac{n}{n-1}  + p(k) \frac2{n-1}). $$

Using $p(k)= \frac{\sqrt{2}}{2} + \frac{ Y(k)}{\sqrt n}$, 
we get
\begin{eqnarray*}
E[ Y(k+1)- Y(k) | Y(k)] = \tfrac{\sqrt 2 - 1}{\sqrt n (n-1)} + Y(k) (- \tfrac{2 \sqrt 2}{n-1} +\tfrac{2}{n(n-1)}) + Y(k)^2 (- \tfrac{2}{\sqrt n(n-1)})
\end{eqnarray*}
which yields the equivalent
$$n E[ Y(k+1)- Y(k) | Y(k)] \approx - {2 \sqrt 2}\ Y(k) $$
when $n$ goes to infinity. We have
$$E[ (Y(k+1)- Y(k))^2 | Y(k)] = n (E[ (p(k+1)- p(k))^2 | p(k)]),$$
hence, since $(p(k+1)- p(k))^2$ is always $\frac1{n^2}$,
$$n E[ (Y(k+1)- Y(k))^2 | Y(k)] = 1.$$

Set $a(x)= -2 \sqrt 2 x$, $b(x) = 1$. From the above calculations we have clearly 
$$\lim_ {n\to \infty} sup_{|x| \le R} |a^{(n)}(x)-a(x)| = 0$$
$$\lim_ {n\to \infty} sup_{|x| \le R} |b^{(n)}(x)-b(x)| = 0$$
for all $R < +\infty$. Since the jumps of $Y^{(n)}$ are bounded in absolute value by $\frac1{\sqrt n}$, $\Delta_\epsilon^{(n)}$ is null, as soon as $\frac1{\sqrt n}$ is smaller than $\epsilon$, and so
$$\lim_{n\to \infty} sup_{|x| \le R} \Delta_\epsilon^{(n)} = 0, \forall \epsilon >0$$

$$\sup_{|x| \le R} K^{(n)}(x) < \infty$$ is still easy to establish.

Now the stochastic differential equation 
\begin{equation} \label{stocasticeq}
dX(t)  = -2 \sqrt 2 X(t) dt + dB(t)
\end{equation}
is of a well-known type. Actually, for $b>0$ and $\sigma\neq 0$, stochastic differential equations of type 
$$dX(t)  = -b X(t) dt + \sigma dB(t)$$
are known to have a unique solution for all initial conditions 
$X(0)=x$.  This solution is an Orstein-Uhlenbeck process, and is given by (see e.g. \cite{comets2006csm})
$$X(t) = e^{-b t} X(0) + \int_{0}^t e^{-b(t-s)} \sigma dB(s).$$

For all initial conditions
$X(0)$, $X(t)$ is known to converge in law, when $t$ goes to infinity, to the
Gaussian distribution $\mN(0,\frac{\sigma^2}{2b})$. This latter Gaussian distribution is
an invariant distribution for the Orstein-Uhlenbeck process. See for example \cite{comets2006csm}.

We have all the ingredients to apply Theorem \ref{th:theo1} again. We obtain:


\begin{theorem}
For all $t$, $$p(\lfloor nt \rfloor) = \frac{\sqrt2}{2} +
  \frac{1}{\sqrt n} A_n(t),$$ where $(A_n(t))_{t\ge 0}$ converges in law to the
  unique solution of stochastic differential equation \eqref{stocasticeq}. Hence $A_n(t)$ converges to 
 $\mN(0,\frac{\sqrt 2}{8})$ when $t$ and $n$ go to infinity.
\end{theorem}

\section{Some Other Algebraic Numbers}\label{sec: otheralg}
We have treated in detail the case of $\sqrt{2}/2$. 
We present in this Section, without proofs, the extension of our result
to 2-states protocols, with pairing of agents.

At each time step, two agents are picked and their states are possibly
changed, according to the fixed rule $\delta$. 
These protocols are completely 
described by the three mean increments
\begin{align*}
\alpha&=\mathbb{E}[n_+(k+1)-n_+(k)|\{+,+\} \mbox{have been picked}],\\
\beta&=\mathbb{E}[n_+(k+1)-n_+(k)|\{+,-\} \mbox{have been picked}],\\
\gamma&=\mathbb{E}[n_+(k+1)-n_+(k)|\{-,-\} \mbox{have been picked}].
\end{align*}
Thus, there are $27=3^3$ different rules. We denote them by the corresponding 
triplet $(\alpha,\beta,\gamma)$. For instance, the rule computing $\sqrt{2}/2$
is denoted by  $(-1,+1,+1)$. We exclude the \emph{identity} rule $(0,0,0)$.

We also set
\begin{align*}
a&=\alpha-2\beta+\gamma, \\
b&=2\beta-2\gamma,\\
c&=\gamma.\\
\end{align*}
Then, we associate to each triplet $(\alpha,\beta,\gamma)$ the polynomial
$$
P=aX^2+bX +c.
$$
The following Lemma, whose proof is omitted, is the basis of the next discussion.
\newtheorem{lem}{Lemma}
\begin{lem}
The polynomial $P=aX^2+bX +c$ admits at most one root in $(0,1)$, which we denote by
$p^\star$. Moreover,
\[
q:=(\alpha^2-2\beta^2 +\gamma^2)(p^\star)^2+(2\beta^2-2\gamma^2)p^\star+\gamma^2>0\]
and
\[
2ap^\star +b<0.
\]
\end{lem}
We will see that the computational power of a protocol population reads on the
corresponding polynomial $P$.

\vspace{8mm}

\noindent{\bf Case 1: $P$ has no root in (0,1). Monotonic convergence. }
\newcommand{\set}[1]{\left\{#1\right\}}

For $10$ rules, $P$ does not admit a root in $(0,1)$. In this case, the convergence of the
corresponding protocol is easy to establish. Take for instance $(0,1,2)$ :
\begin{align*}
\begin{cases}
++&\mapsto ++\\
+-&\mapsto ++\\
--&\mapsto ++
\end{cases}
\end{align*}
It is clear that the protocol converges to the configuration $\set{+}^n$. 
We summarize the behaviors of the $9$ remaining rules in the following table.
\newcommand{\ph}{\phantom}
\begin{center}
\begin{tabular}{| p{8mm} |  p{8mm} |  p{8mm} | p{75mm} |}
\hline
$\ph{-}\alpha$ & $\ph{-}\beta$ & $\ph{-}\gamma$ & Convergence\\
\hline
\hline
 $\ph{-}0$     &   $\ph{-}1$   &   $\ph{-}0$    &   $\set{+}^n$ (or $\set{-}^n$ if it is the init. config.)\\  
 $\ph{-}0$     &   $\ph{-}1$   &   $\ph{-}1$    &   $\set{+}^n$    \\  
 $\ph{-}0$     &   $\ph{-}1$   &   $\ph{-}2$    &   $\set{+}^n$    \\  
 $\ph{-}0$     &   $\ph{-}0$   &   $\ph{-}1$    &   $\set{-}\set{+}^{n-1}$    \\  
 $\ph{-}0$     &   $\ph{-}0$   &   $\ph{-}2$    &   $\set{+}^n$ or $\set{-}\set{+}^{n-1}$     \\  
 $\ph{-}0$     &   $-1$   &   $\ph{-}0$    &   $\set{-}^n$  (or $\set{+}^n$  if it is the init. config.) \\  
 $-1$     &   $\ph{-}0$   &   $\ph{-}0$    &   $\set{+}\set{-}^{n-1}$    \\  
 $-1$     &   $-1$   &   $\ph{-}0$    &   $\set{-}^n$    \\  
 $-2$     &   $\ph{-}0$   &   $\ph{-}0$    &   $\set{-}^n$ ou  $\set{+}\set{-}^{n-1}$   \\  
 $-2$     &   $-1$   &   $\ph{-}0$    &   $\set{-}^n$    \\  
\hline
\end{tabular}
\end{center}

\vspace{8mm}

\noindent{\bf Case 2: $P$ has a unique root in (0,1). Approximation with a diffusion. }
Depending on $(\alpha,\beta,\gamma)$, $p^\star$ has one of the three 
following expressions:
$$
p^\star=
\begin{cases}
&\frac{-b+\sqrt{b^2-4ac}}{2a},\\
&\frac{-b-\sqrt{b^2-4ac}}{2a},\\
&-\frac{c}{b}.
\end{cases}
$$
As in the previous section, we set
\begin{displaymath}
Y_k=Y^{(n)}_k:=\sqrt{n}(p^{(n)}_k- p^\star),
\end{displaymath}
and $X$ is the linear interpolation of $Y$:
$$
X^{(n)}(t) = Y^{(n)}_{\lfloor nt \rfloor}  + (nt - \lfloor nt \rfloor) (Y^{(n)}_{\lfloor nt+1 \rfloor} - Y^{(n)}_{\lfloor nt \rfloor}).
$$
\begin{theorem}\label{Th:EDS}
Assume that $p_0^{(n)}$ converges to a random variable $X_0$ in $(0,1)$.
When $n$ goes to infinity, the process $\left(X^{(n)}(t)\right)_{t\geq 0}$
converges to the unique (weak) solution $X$ of the stochastic differential equation
\begin{equation}\label{Eq:EDS}
dX_t=(2ap^\star +b)X_t dt +qdB_t
\end{equation}
starting at $X_{0}$. Since, by the previous Lemma, $2ap^\star +b<0$ and $q>0$,
this solution $X$ is an Ornstein-Uhlenbeck process, and has the following representation:
\begin{displaymath}
X_t=X_0e^{(2ap^\star +b)t} +q\int_0^t \exp^{(2ap^\star +b)(t-s)}dB_s.
\end{displaymath}
\end{theorem}
\begin{proof}
We omit the proof, for it extends easily from the case $\sqrt{2}/2$.
\end{proof}
There are $16$ rules for which $P$ has a root in $(0,1)$. They compute $13$ different
algebraic numbers, as shown in the next table.
\begin{center}
\newcommand{\Poly}[3]{
#1X^2+#2X+#3
}
\begin{tabular}{| p{8mm} |  p{8mm} |  p{8mm} | r | l |}
\hline
$\ph{-}\alpha$ & $\ph{-}\beta$ & $\ph{-}\gamma$ & Polynomial $P$ & $p^\star$ \\
\hline
\hline
 $\ph{-}0$     &   $-1$   &   $\ph{-}1$    &  $3X^2-4X+1$      & $1/3$ \\  
 $\ph{-}0$     &   $-1$   &   $\ph{-}2$    &  $4X^2-6X+2$      & $1/2$ \\  
 $-1$     &   $\ph{-}1$   &   $\ph{-}0$    &  $-3X^2+2X$           & $2/3$ \\  
 $-1$     &   $\ph{-}1$   &   $\ph{-}1$    &   $-2X^2+1$          & $\sqrt{2}/2$ \\  
 $-1$     &   $\ph{-}1$   &   $\ph{-}2$    &   $-X^2-2X+2$      & $\sqrt{3}-1$ \\  
 $-1$     &   $\ph{-}0$   &   $\ph{-}1$    &   $-2X+1$          & $1/2$ \\  
 $-1$     &   $\ph{-}0$   &   $\ph{-}2$    &   $X^2-4X+2$         &  $2-\sqrt{2}$\\  
 $-1$     &   $-1$   &   $\ph{-}1$    &   $2X^2-4X+1$    & $1-\sqrt{2}/{2}$\\  
 $-1$     &   $-1$   &   $\ph{-}2$    &   $3X^2-6X+2$        & $1-\sqrt{3}/{3}$\\  
 $-2$     &   $\ph{-}1$   &   $\ph{-}0$    &   $-4X^2+2X$           & $1/2$\\  
 $-2$     &   $\ph{-}1$   &   $\ph{-}1$    &   $-3X^2+1$           & $\sqrt{3}/{3}$\\  
 $-2$     &   $\ph{-}1$   &   $\ph{-}2$    &   $-2X^2-2X+2$          & $(\sqrt{5}-1)/2$ \\  
 $-2$     &   $\ph{-}0$   &   $\ph{-}1$    &   $-X^2-2X+1$         &  $\sqrt{2}-1$ \\  
 $-2$     &   $\ph{-}0$   &   $\ph{-}2$    &   $-4X+2$          & $1/2$\\  
 $-2$     &   $-1$   &   $\ph{-}1$    &   $X^2-4X+1$          & $2-\sqrt{3}$ \\  
 $-2$     &   $-1$   &   $\ph{-}2$    &   $2X^2-6X+2$         &  $(3-\sqrt{5})/2$ \\  
\hline
\end{tabular}
\end{center}

\section{Conclusion}
\label{sec:concl}

Population protocols have been introduced in
\cite{AspnesADFP2004} as a model for sensor networks. 
In this paper we considered population protocols with a huge
population hypothesis. 

Whereas for the definitions of computability considered in
\cite{AspnesADFP2004}, some population protocols are not considered as
(stably) convergent, we proved through an example that sometimes they
actually compute in some natural sense some irrational algebraic
values: indeed, we gave a simple example where the proportion of agents
in state $+$ converges to $\frac{\sqrt 2}{2}$, whatever the initial
state of the system is.

One aim of this paper was to formalize the proof of convergence. We
did it on this example and on general population protocols using a
diffusion approximation technique, using a theorem due to
\cite{stroock1979mdp}. 

We consider this work as a first step towards understanding which
numbers can be computed by such protocols in this sense. We gave some
preliminary results about numbers computable with $2$-states
protocols.

Considering more states, it is easy to derive from the protocol
considered here another protocol that would compute $\sqrt {\sqrt
  \frac12}$, by working with an alphabet made of pairs of
states. Whereas it is easy to see that computable numbers in this
sense must be algebraic numbers of $[0,1]$, we didn't succeed yet to
characterize precisely computable numbers.


\begin{thebibliography}{10}

\bibitem{AspnesADFP2004}
Dana Angluin, James Aspnes, Zo{\"e} Diamadi, Michael~J. Fischer, and Ren\'e
  Peralta.
\newblock Computation in networks of passively mobile finite-state sensors.
\newblock In {\em Twenty-Third ACM Symposium on Principles of Distributed
  Computing}, pages 290--299. ACM Press, July 2004.

\bibitem{AngluinAE2006semilinear}
Dana Angluin, James Aspnes, and David Eisenstat.
\newblock Stably computable predicates are semilinear.
\newblock In {\em PODC '06: Proceedings of the twenty-fifth annual ACM
  symposium on Principles of distributed computing}, pages 292--299, New York,
  NY, USA, 2006. ACM Press.

\bibitem{PopProtocolsEATCS}
James Aspnes and Eric Ruppert.
\newblock An introduction to population protocols.
\newblock In {\em Bulletin of the EATCS}, volume~93, pages 106--125, 2007.

\bibitem{ChatzigiannakisS08}
Ioannis Chatzigiannakis and Paul~G. Spirakis.
\newblock The dynamics of probabilistic population protocols.
\newblock In {\em Distributed Computing, 22nd International Symposium, DISC},
  volume 5218 of {\em Lecture Notes in Computer Science}, pages 498--499. 2008.

\bibitem{ChatzigiannakisS08rr}
Ioannis Chatzigiannakis and Paul~G. Spirakis.
\newblock The dynamics of probabilistic population protocols.
\newblock Technical report, arXiv:0807.0140v1, 2008.

\bibitem{comets2006csm}
F.~Comets and T.~Meyre.
\newblock {\em {Calcul stochastique et modeles de diffusions}}.
\newblock Dunod Paris, 2006.

\bibitem{diamadi2001sgs}
Z.~Diamadi and M.J. Fischer.
\newblock {A simple game for the study of trust in distributed systems}.
\newblock {\em Wuhan University Journal of Natural Sciences}, 6(1-2):72--82,
  2001.

\bibitem{durrett1996scp}
R.~Durrett.
\newblock {\em {Stochastic Calculus: A Practical Introduction}}.
\newblock CRC Press, 1996.

\bibitem{givon2004emd}
D.~Givon, R.~Kupferman, and A.~Stuart.
\newblock {Extracting macroscopic dynamics: model problems and algorithms}.
\newblock {\em Nonlinearity}, 17(6):R55--R127, 2004.

\bibitem{Les06}
Annick Lesne.
\newblock Discrete vs continuous controversy in physics.
\newblock {\em Mathematical Structures in Computer Science}, 2006.
\newblock In print.

\bibitem{presburger:uvk}
M.~Presburger.
\newblock {Uber die Vollstandig-keit eines gewissen systems der Arithmetik
  ganzer Zahlen, in welchemdie Addition als einzige Operation hervortritt}.
\newblock {\em Comptes-rendus du I Congres des Mathematicians des Pays Slaves},
  pages 92--101, 1929.

\bibitem{stroock1979mdp}
D.W. Stroock and SRS Varadhan.
\newblock {\em {Multidimensional Diffusion Processes}}.
\newblock Springer, 1979.

\end{thebibliography}
\end{document}